\documentclass[12pt]{article}
\usepackage{color,epsfig,citesort,amsmath,fullpage}
\usepackage{algorithm}
\usepackage{algorithmic}

\def\myendproof{{\hfill \vbox{\hrule\hbox{%
   \vrule height1.3ex\hskip0.8ex\vrule}\hrule }}\par}
\newtheorem{theorem}{Theorem}
\newtheorem{lemma}[theorem]{Lemma}

\newenvironment{proof}{{\it Proof. }}{\myendproof}

\newcommand{\setof}[1]{\{{#1}\}}
\newcommand{\Xomit}[1]{}

\title{{\bf Linear-Time Algorithms for the Paired-Domination Problem in Interval Graphs and Circular-Arc Graphs}}
\author{Ching-Chi Lin\thanks{Department of Computer Science and Engineering,
                             National Taiwan Ocean University,
                             Keelung 20224, Taiwan. Corresponding author.
                             Email: lincc@mail.ntou.edu.tw}
        \and
        Hai-Lun Tu\thanks{Department of Computer Science and
                          Information Engineering,
                          National Taiwan University,
                          Taipei 10617, Taiwan.
                          Email:d95019@csie.ntu.edu.tw}
        }

\date{\today}

\begin{document}
\maketitle

\begin{abstract}
In a graph $G$, a vertex subset $S\subseteq V(G)$ is said to be a
dominating set of $G$ if every vertex not in $S$ is adjacent to a
vertex in $S$. A dominating set $S$ of a graph $G$ is called a
paired-dominating set if the induced subgraph $G[S]$ contains a
perfect matching. The paired-domination problem involves finding a
smallest paired-dominating set of $G$. Given an intersection model
of an interval graph $G$ with sorted endpoints, Cheng~{\em et
al.}~\cite{CKN07} designed an $O(m+n)$-time algorithm for interval
graphs and an $O(m(m+n))$-time algorithm for circular-arc graphs.
In this paper, to solve the paired-domination problem in interval
graphs, we propose an $O(n)$-time algorithm that searches for a
minimum paired-dominating set of $G$ incrementally in a greedy
manner. Then, we extend the results to design an algorithm for
circular-arc graphs that also runs in $O(n)$ time.

\bigskip


\end{abstract}

\newpage

\def\skippt{23pt}
\baselineskip \skippt

\section{Introduction}\label{section:intro}
The museum protection problem can be accurately represented by a
graph $G = (V(G),E(G))$. The vertex set of $G$, denoted by $V(G)$,
represents the sites to be protected; and the edge set of $G$,
denoted by $E(G)$, represents the set of protection capabilities.
There exists an edge $xy$ connecting vertices $x$ and $y$ if a
guard at site $x$ is capable of protecting site $y$ and vice
versa. In the classical domination problem, it is necessary to
minimize the number of guards such that each site has a guard or
is in the protection range of some guard. For the
paired-domination problem, in addition to protecting the sites,
the guards must be able to back each other up~\cite{HS98}.
Throughout this paper, we let $n = |V(G)|$ and $m = |E(G)|$.

In a graph $G$, a vertex subset $S\subseteq V(G)$ is said to be a
{\em dominating set} of $G$ if every vertex not in $S$ is adjacent
to a vertex in $S$. A dominating set $S$ of a graph $G$ is called
a paired-dominating set if the induced subgraph $G[S]$ contains a
perfect matching. The paired-domination problem involves finding a
smallest paired-dominating set of $G$. Haynes and
Slater~\cite{HS98} defined the paired-domination problem and
showed that it is NP-complete in general graphs. More recently,
Chen~{\em et al.}~\cite{Chen10} demonstrated that the problem is
also NP-complete in bipartite graphs, chordal graphs, and split
graphs. Panda and Pradhan~\cite{Panda12} strengthened the above
results by showing that the problem is NP-complete for perfect
elimination bipartite graphs. In addition, McCoy and
Henning~\cite{McCoy09} investigated variants of the
paired-domination problem in graphs.


Meanwhile, several polynomial-time algorithms have been developed
for some special classes of graphs such as tree graphs, interval
graphs, strongly chordal graphs, and circular-arc graphs.
Qiao~{\em et al.}~\cite{QKCD03} proposed an $O(n)$-time algorithm
for tree graphs; Kang~{\em et al.}~\cite{KSC04} presented an
$O(n)$-time algorithm for inflated trees; Chen~{\em et
al.}~\cite{Chen09} designed an $O(m+n)$-time algorithm for
strongly chordal graphs; and Cheng~{\em et al.}~\cite{CKS09}
developed an $O(mn)$-time algorithm for permutation graphs. To
improve the results in~\cite{CKS09}, Lappas~{\em et
al.}~\cite{Lappas09} introduced an $O(n)$-time algorithm. In
addition, Hung~\cite{Hung12} described an $O(n)$-time algorithm
for convex bipartite graphs; Panda and Pradhan~\cite{Panda12}
proposed an $O(m+n)$-time algorithm for chordal bipartite graphs;
Chen~{\em et al.}~\cite{Chen10} introduced $O(m+n)$-time
algorithms for block graphs and interval graphs; and Cheng~{\em et
al.}~\cite{CKN07} designed an $O(m+n)$-time algorithm for interval
graphs and an $O(m(m+n))$-time algorithm for circular-arc graphs.
In this paper, given an intersection model of interval graph $G$
with sorted endpoints, we improve the above results with time
complexity $O(n)$ for interval graphs and circular-arc graphs.

Several variants of the classic domination problem, such as the
weighted domination, edge domination, independent domination,
connected domination, locating domination, and total domination
problems, have generated a great deal of research interest in
recent decades~\cite{Chang04,Haynes98,Haynes98-2,Hedetniemi91}. It
has been proved that the above problems are NP-complete in general
graphs but they yield polynomial-time results in some special
classes of
graphs~\cite{Bertossi88,D'Atri88,Chang97,Lu02,Ramalingam88,Chang98}.
In particular, these variants have been studied  intensively in
interval and circular-arc graphs~\cite{Ramalingam88,Chang98}.

For weighted interval graphs, Ramalingam and
Rangan~\cite{Ramalingam88} proposed a unified approach to solve
the independent domination, domination, total domination and
connected domination problems in $O(n+m)$ time. Subsequently,
Chang~\cite{Chang98} developed an $O(n)$-time algorithm for the
independent domination problem, an $O(n)$-time algorithm for the
connected domination problem and an $O(n \log \log n)$-time
algorithm for the total domination problem in weighted interval
graphs. The author also extended the results to derive
$O(m+n)$-time algorithms for the same problems in circular-arc
graphs. Moreover, Hsu and Tsai~\cite{HT91} developed an
$O(n)$-time algorithm for the classic domination problem in
circular-arc graphs. The algorithm, which utilizes a greedy
strategy, motivated the algorithms proposed in this paper. Note
that all the above algorithms assume that an intersection model of
$G$ with sorted endpoints is given.

In this paper, we show that the paired-domination in interval
graphs and circular-arc graphs is solvable in linear time. More
precisely, given an intersection model of a circular-arc graph $G$
with sorted endpoints, we propose an $O(n)$-time algorithm that
produces a minimum paired-dominating set of $G$. Moreover, because
circular-arc graphs are a natural generalization of interval
graphs, the paired-domination problem in the latter can also be
solved in $O(n)$ time.

The remainder of this paper is organized as follows. In
Section~\ref{section:algo-interval}, we introdece an $O(n)$-time
algorithm for interval graphs; and in
Section~\ref{section:algo-cir}, we extend the result to derive an
$O(n)$-time algorithm for circular-arc graphs.
Section~\ref{section:conclusion} contains some concluding remarks.

\section{The Proposed Algorithm for Interval Graphs
                              \label{section:algo-interval}}
To find a minimum paired-dominating set of an interval graph $G$,
we designed an $O(n)$-time algorithm that derives the set
incrementally in a greedy manner. Before describing the approach
in detail, we introduce some preliminaries for interval graphs.

\bigskip
\begin{figure}[thb]
\centerline{\input{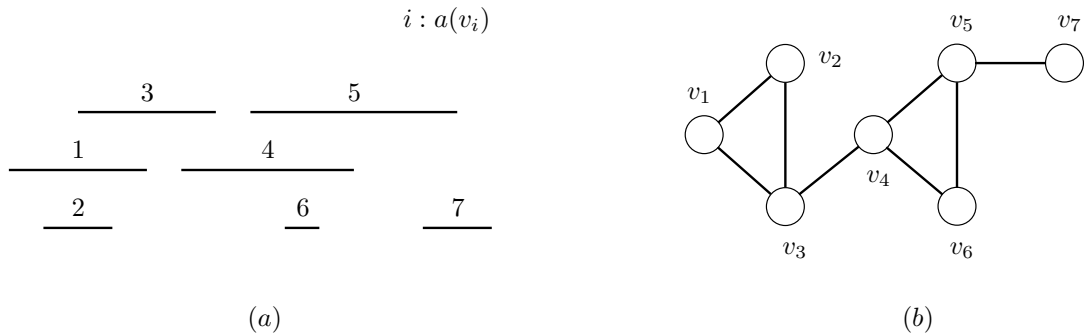}} \caption{$(a)$ A family of intervals on
a real line. $(b)$ The corresponding interval graph $G$ for the
family of intervals in $(a)$.} \label{fig:1}
\end{figure}

A graph $G$ is deemed an {\em interval graph} if there is a
one-to-one correspondence between its vertices and a family of
intervals, $I$, on a real line, such that two vertices in the
graph have an edge between them if and only if their corresponding
intervals overlap. Interval graphs have received considerable
attention because of their application in the real world. Booth
and Lueker~\cite{Booth76} designed an algorithm that can recognize
interval graphs in $O(n + m)$ time. As a byproduct, an
intersection model $I$ of an interval path graph $G$ can be
constructed in $O(n + m)$ time. In the remainder of this section,
we assume that $G$ is an interval graph with $V(G)=\setof{v_1,
v_2, \ldots, v_n}$, where $n\ge 3$. We also assume that an
intersection model $I$ is available to $G$, as shown by the examle
in Figure~\ref{fig:1}, where Figure~\ref{fig:1}$(b)$ depicts the
corresponding interval graph $G$ for the family of intervals in
Figure~\ref{fig:1}$(a)$.

The {\em neighborhood} $N_G(v)$ of a vertex $v$ is the set of all
vertices adjacent to $v$ in $G$; and the {\em closed neighborhood}
$N_G[v] = \{v\} \cup N_G(v)$. For each $v\in V(G)$, let $a(v)$
denote the corresponding interval of $v$ in $I$. Each interval is
represented by $[\ell(v), r(v)]$, where $\ell(v)$ and $r(v)$ are,
respectively, the {\em left endpoint} and the {\em right endpoint}
of $a(v)$. It is assumed that the left endpoint $\ell(v_i)$ is on
the left of the left endpoint $\ell(v_{i+1})$ for all $ 1 \le i
\le n -1$. Without loss of generality, we assume that all interval
endpoints (i.e., $\ell(v)$ and $r(v)$) are distinct. In addition,
each interval endpoint is assigned a positive integer between $1$
and $2n$ in ascending order in a left-to-right traversal. For any
two endpoints $x$ and $y$, $x$ is said to be lower than $y$,
denoted by ``$x < y$", if its label is lower than the label of
$y$, i.e., $x$ lies on the left of $y$ in $I$; otherwise, $x$ is
said to be greater than or equal to $y$, denoted by ``$x \ge y$".

\subsection{The algorithm
                              \label{section:outline-algo-int}}
As mentioned earlier, our algorithm for finding a minimum
paired-dominating set of an interval graph utilizes a greedy
strategy. With $O(n)$-time preprocessing, the algorithm traverses
the intersection mode $I$ of an interval graph $G$ from left to
right exactly once. For a subset $S\subseteq V(G)$, a vertex $v$
in $G$ is said to be the {\em next undominated} vertex with
respect to $S$ if $a(v)$ has the leftmost right endpoint among the
corresponding intervals of the vertices that are not in $S$ or
adjacent to any vertex in $S$. For each $v \in V(G)$, let the {\em
partner} of $v$, denoted by $p(v)$, be the neighbor of $v$ such
that $a(p(v))$ has the rightmost right endpoint in $I$ among the
corresponding intervals of $N_G(v)$. For the example of
Figure~\ref{fig:1}, we have $p(v_3) = v_{4}$, $p(v_5) = v_{7}$,
and $p(v_7) = v_{5}$. Initially, we set $S=\emptyset$. Then, the
algorithm iteratively finds the next undominated vertex $v$ with
respect to $S$ and adds a pair $(p(v), p(p(v))$ to $S$ until every
vertex not in $S$ is adjacent to a vertex in $S$. The steps of the
algorithm are detailed in Algorithm 1.

\vspace{0.5pt}
\begin{algorithm}
\caption{Finding a minimum paired-dominating set in an interval
graph}
\begin{algorithmic} [1]
\REQUIRE An intersection model $I$ of an interval graph
$G$ with sorted endpoints.
\ENSURE A minimum paired-dominating set $S$ of
$G$.
\STATE let $S \leftarrow \emptyset$;
\REPEAT
\STATE find the next undominated vertex $v$ with
respect to $S$;
\STATE let $S \leftarrow S \cup \setof{p(v),p(p(v))}$;
\UNTIL every vertex not in $S$ is adjacent to a vertex
in $S$
\RETURN $S$;
\end{algorithmic}
\end{algorithm}

For the example in Figure~\ref{fig:1}, the algorithm generates a
minimum paired-dominating set $S = \{v_3,v_4,v_5,v_7\}$ of $G$. In
the following, we demonstrate the correctness of the algorithm;
and then describe an $O(n)$-time implementation of the algorithm
in the next subsection. First, we introduce some necessary
notations. Let $G[S]$ denote the subgraph of $G$ induced by a
subset $S$ of $V(G)$; and let $\tilde{N}_{G_{i}}(v_x)$ represent
the set of all vertices adjacent to $v_x$ in $G[\setof{v_x, v_i,
v_{i+1}, \ldots, v_n}]$. In addition, let $v_{i^*}$ be the vertex
in $V(G)$ such that $\ell(v_{i^*})$ is the first left endpoint
encountered in a left-to-right traversal from $r(v_i)$ in $I$. For
the example in Figure~\ref{fig:1}, we have $v_{2^*} = v_{4}$,
$v_{3^*} = v_{5}$, $\tilde{N}_{G_{2^*}}(v_3) = \setof{v_4}$ and
$\tilde{N}_{G_{3^*}}(v_4) = \setof{v_5,v_6}$.

\vspace{-2pt}

\begin{lemma}\label{lemma:partial_dom}
If $v_x\in N_G(v_i)$ and $v_y \in N_G(v_x)$, we have
\\ $~\hspace{2.1cm} \setof{\tilde{N}_{G_{i^*}}(v_x) \cup \tilde{N}_{G_{i^*}}(v_y)} \subseteq
\setof{\tilde{N}_{G_{i^*}}(p(v_i)) \cup \tilde{N}_{G_{i^*}}(p(p(v_i)))}$ \quad for $1 \le i \le n$.
\end{lemma}

\vspace{-2pt}

\begin{proof}
By the definition of $p(v_i)$, we have
$\text{max}\setof{r(p(v_i)), r(p(p(v_i)))} \ge
\text{max}\setof{r(v_x), r(v_y)}$. For the case where
$r(p(p(v_i)))
> r(p(v_i))$, the segment $[r(v_i), r(p(p(v_i)))]$ is contained
in $a(p(v_i)) \cup a(p(p(v_i)))$. This implies that if $v_j \in
\tilde{N}_{G_{i^*}}(v_x) \cup \tilde{N}_{G_{i^*}}(v_y)$, we also
have $v_j\in \tilde{N}_{G_{i^*}}(p(v_i)) \cup
\tilde{N}_{G_{i^*}}(p(p(v_i)))$. The arguments are similar in the
case where $r(p(p(v_i))) < r(p(v_i))$.
\end{proof}

\vspace{-2pt}

\begin{lemma} \label{lemma:interval-correctness}
Given an intersection model $I$ of an interval graph $G$ with
sorted endpoints, Algorithm $1$ outputs a minimum
paired-dominating set $S$ of $G$.
\end{lemma}

\vspace{-2pt}

\begin{proof}
Suppose that the algorithm outputs $S=\{v_{s_1}, v_{s_{1'}},
\ldots, v_{s_x},v_{s_{x'}}\}$, where $v_{s_i}$ and $v_{s_{i'}}$
are added to $S$ in the $i$th iteration. Clearly, $S$ is a
paired-dominating set of $G$. We prove that $S$ is a minimum
paired-dominating set of $G$ as follows. Let $Z =\{v_{z_1},
v_{z_{1'}}, \ldots, v_{z_y},v_{z_{y'}}\}$ be a minimum
paired-dominating set of $G$ such that $(v_{z_1}, v_{z_{1'}}),
\ldots, (v_{z_y},v_{z_{y'}})$ is a perfect matching of $G[Z]$ and
$z_i < z_{i+1}$ for $1 \le i \le y -1$. In addition, let $S_i =
\{v_{s_1}, v_{s_{1'}},\ldots, v_{s_i},v_{s_{i'}}\}$ and $Z_i =
\{v_{z_1}, v_{z_{1'}},\ldots, v_{z_i},v_{z_{i'}}\}$. For a subset
$R \subseteq V(G)$, we define that $A(R) = \max\{i \mid v_i \in R$
$\text{~or~} v_i \text{~is adjacent to a vertex~} v_j\in R\}$. To
prove that $S$ is a minimum paired-dominating set of $G$, it is
sufficient to show that $A(S_i) \ge A(Z_i)$ for $1 \le i \le y$.

We prove the above statement by induction on $i$. By the
definitions of the next undominated vertex and $p(v)$, the
statement holds for $i = 1$, and we assume the statement holds for
$i =k$. Consider the case where $i = k + 1$. Let $v_j$ be the next
undominated vertex with respect to $\{v_{s_1}, v_{s_{1'}}, \ldots,
v_{s_k},v_{s_{k'}}\}$. Clearly, $j
> A(S_k) \ge A(Z_k)$. First, we consider the case where $v_j \not \in N_G(v_{z_{(k+1)}})
\cup N_G(v_{z_{(k+1)'}})$. By the definition of $v_j$, we have
$A(S_{k+1}) \ge j
> A(Z_{k+1})$. Next, we consider the case where $v_j \in
N_G(v_{z_{(k+1)}}) \cup N_G(v_{z_{(k+1)'}})$. According to
Lemma~\ref{lemma:partial_dom}, we have
$\setof{\tilde{N}_{G_{j^*}}(v_{z_{(k+1)}}) \cup
\tilde{N}_{G_{j^*}}(v_{z_{(k+1)'}})} \subseteq
\setof{\tilde{N}_{G_{j^*}}(v_{s_{(k+1)}}) \cup
\tilde{N}_{G_{j^*}}(v_{s_{(k+1)'}})}$. It follows that $A(S_{k+1})
\ge A(Z_{k+1})$ in both cases, so $x \le y$. By the definition of
$Z$, we have $y \le x$; thus, $x = y$. The lemma then follows.
\end{proof}

\subsection{An $O(n)$-time implementation of Algorithm $1$
                              \label{section:implement-int}}
For each $v \in V(G)$, let $next(v)$ denote the vertex in
$V(G)-N_G[v]$ whose corresponding right endpoint is the first
right endpoint encountered in a left-to-right traversal from
$r(v)$ in $I$. The algorithm traverses the intersection mode $I$
of an interval graph $G$ from left to right exactly once.
Therefore, to prove that the algorithm runs in $O(n)$ time, it
suffices to show that, with preprocessing in $O(n)$ time, the
algorithm takes $O(1)$ time to determine $p(v_i)$ and $next(v_i)$
for $1 \le i \le n$. In the following, we describe two $O(n)$-time
preprocessing procedures used to determine all $p(v_i)$ and
$next(v_i)$, respectively.

To obtain $p(v_i)$, we traverse the endpoints of $I$ from left to
right and maintain a variable $rightmost$ that represents the
interval containing the rightmost endpoint in the current stage.
Initially, we set $rightmost = v_1$. When a left endpoint
$\ell(v_i)$ is visited, we compare $r(v_i)$ with $r(rightmost)$.
If $r(v_i) > r(rightmost)$, we set $rightmost = v_i$; otherwise,
we do nothing. When a right endpoint $r(v_i)$ is visited, we set
$p(v_i) = rightmost$. Because there are $2n$ endpoints in $I$, the
procedure can be completed in $O(n)$ time.

Next, we describe the $O(n)$-time procedure used to determine
$next(v_i)$. Let $f(v_i)$ be the vertex in $V(G)$ such that
$\ell(f(v_i))$ is the first left endpoint encountered in a
left-to-right traversal from $r(v_i)$ in $I$. In addition, let
$c(v_i) = v_j$ be the vertex in $V(G)$ such that $r(v_j)$ is the
first right endpoint encountered in a left-to-right traversal from
$\ell(v_i)$ in $I$ with $j \ge i$. Then, $next(v_i)$ can be seen
as the composite function of $c$ and $f$, i.e., $next(v_i) =
c(f(v_i))$. The following discussion shows that, with $O(n)$-time
preprocessing, $c(v_i)$ and $f(v_i)$ can be determined in $O(1)$
time.

First, we describe an $O(n)$-time procedure to determine $f(v_i)$
for each $v_i \in V(G)$. Again, we traverse the endpoints of $I$
from left to right and maintain a set $P$. Initially, we set $P =
\emptyset$. When a right endpoint $r(v_i)$ is visited, we add
$v_i$ to $P$; and when a left endpoint $\ell(v_i)$ is visited, we
set $f(v_j) = v_i$ for each vertex $v_j$ in $P$ and set $P$ to be
empty. The arguments for determining $c(v_i)$ are similar. If we
reach a left endpoint $\ell(v_i)$, we insert $v_i$ into the queue
$Q$. When we find a right endpoint $r(v_i)$, we set $c(v_j)=v_i$
and remove $v_j$ from the queue for all $v_j$ with $j \le i$.
Because there are $2n$ endpoints in $I$, the above two procedures
can also be completed in $O(n)$ time.

Combining Lemma~\ref{lemma:interval-correctness} and above
discussion, we have the following theorem, which is one of the key
results presented in this paper.

\begin{theorem} \label{theorem:interval} Given an intersection
model $I$ of an interval graph $G$ with sorted endpoints,
Algorithm 1 outputs a minimum paired-dominating set $S$ of $G$ in
$O(n)$ time.
\end{theorem}

\section{The Algorithm for Circular-arc Graphs
                              \label{section:algo-cir}}

In this section, we extend the previous results to derive an
$O(n)$-time algorithm for finding a minimum paired-dominating
set in a circular-arc graph. The algorithm also exploits a greedy strategy. A graph $G$ is deemed a {\em circular-arc graph} if there is a one-to-one correspondence between $V(G)$ and a set of arcs on a circle such that $(u, v)\in E(G)$ if and only if the corresponding arc of $u$ overlaps with the corresponding arc of $v$. An {\em intersection model} of $G$ is a circular ordering of its corresponding arc endpoints when moving in a counterclockwise direction around the circle. McConnell~\cite{McConnell03} proposed an $O(n+m)$-time algorithm that recognize a circular-arc graph $G$, and simultaneously obtains an intersection model of $G$ as a byproduct. In the following discussion, we assume that $(1)$ $G$ is a circular-arc graph such that $V(G) = \setof{v_1,v_2,\ldots,v_n}$ with $n \ge 3$; and $(2)$ the intersection mode $F$ of $G$ is available.

\hspace{1pt}

\begin{figure}[thb]
\centerline{\input{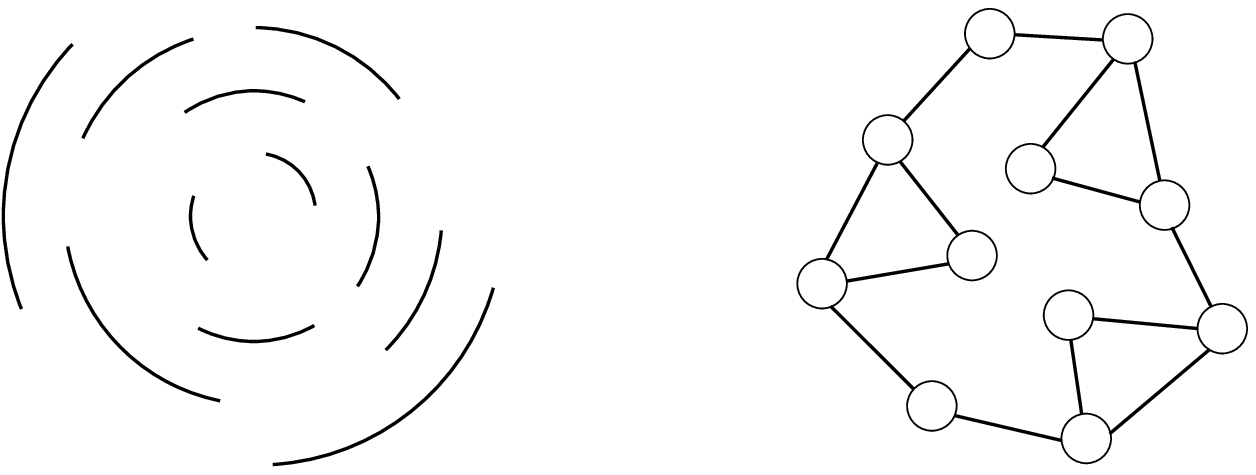}} \caption{$(a)$ A family of arcs on a
circle. $(b)$ The corresponding circular-arc graph $G$ for the
family of arcs in $(a)$.} \label{fig:2}
\end{figure}

For each $v\in V(G)$, let $a(v)$ denote the corresponding arc of
$v$ in $F$. Each arc is represented by $[h(v), t(v)]$, where
$h(v)$ is the {\em head} of $a(v)$, $t(v)$ is the {\em tail} of
$a(v)$, and $h(v)$ precedes $t(v)$ in a clockwise direction.
Moreover, for any subset $W$ of $V(G)$, we define $a(W)=\setof{a(v) \mid v\in W}$. It is assumed that all $h(v)$ and $t(v)$ are distinct and no single arc in $F$ covers the whole circle. All endpoints are assigned positive integers between $1$ and
$2n$ in ascending order in a clockwise direction.

In addition, we assume that $h(v_1) = 1$. We also
assume that $a(v_1)$ is chosen arbitrarily from $F$; and we let
$a(v_2), a(v_3), \ldots, a(v_n)$ be the ordering of arcs in
$F-\{a(v_1)\}$ such that $h(v_i)$ is encountered before $h(v_j)$
in a clockwise direction from $h(v_1)$ if $i<j$.
Figure~\ref{fig:2} shows an illustrative example, in which
Figure~\ref{fig:2}$(b)$ depicts the corresponding circular-arc
graph $G$ for the family of arcs in Figure~\ref{fig:2}$(a)$.
An ordering of the family of arcs is also provided. For
each $v \in V(G)$, let the {\em tail partner} of $v$, denoted by
$p_t(v)$, be the neighbor of $v$ such that $a(p_t(v))$ contains $t(v)$; and $t(p_t(v))$ is the last tail encountered in clockwise direction from $t(v)$ in $F$. Similarly, let the {\em head partner} of $v$, denoted by $p_h(v)$, be the neighbor of $v$ such that $a(p_h(v))$ contains $h(v)$; and $h(p_h(v))$ is the last head encountered in a counterclockwise direction from $h(v)$ in $F$. For the example in Figure~\ref{fig:2}, we have $p_h(v_1) = v_{11}$, $p_t(v_1) = v_2$, $p_h(v_2) = v_{11}$, and $p_t(v_2) = v_3$.

\subsection{The Algorithm
                              \label{section:outline-algo-arc}}
The algorithm for finding a minimum paired-dominating set of a
circular-arc graph $G$ is similar to the algorithm for interval graphs. An arc is {\em maximal} if it is not contained in any other arc of $F$. Suppose $W$ is the set of neighbors, $u$, of $v_1$ such that $a(u)$ is a maximal arc in $F$, i.e., $W = \{u \mid u \in N[v_1] $~and~$ a(u) $~is a maximal arc in~$ F\}$. Then, we can show that there exists a minimum paired-dominating set $S$ of $G$ such that $W \cap S \not = \emptyset$. If $W \cap S \not =
\emptyset$, we are done. Otherwise, let $v_x$ be a vertex in
$N[v_1] \cap S$ and $v_y$ be a vertex in $W$ such that $a(v_y)$
contains $a(v_x)$ in $F$. Clearly, $(S-\setof{v_x})\cup\setof{v_y}$ is also a minimum paired-dominating set of $G$.

Based on the above observation, we designed a two-step algorithm
for circular-arc graphs. First, the algorithm computes a paired-dominating set $S_i$ for each vertex $w_i \in W$, where $S_i$ is a minimum paired-dominating set among all paired-dominating sets that contain $w_i$. Then, a minimum paired-dominating set $S$ of $G$ is chosen from $S_1,\ldots,S_k$ with $k = |W|$. To find $S_i$, the
algorithm traverses the intersection mode $F$ of a circular-arc
graph $G$ in a clockwise direction. Lemma~\ref{lemma:initial-pair} below proves that there exists a minimum
paired-dominating set $S_i$ such that we have $p_t(w_i) \in
S_i$ or $p_h(w_i) \in S_i$. With the aid of the lemma, the
algorithm first computes two paired-dominating sets $S_i^t$ and
$S_i^h$ that contain the vertices $\setof{w_i, p_t(w_i)}$ and
$\setof{w_i, p_h(w_i)}$ respectively. If $|S_i^t| \le |S_i^h|$,
we have $S_i = S_i^t$; otherwise, we have $S_i = S_i^h$.

To explain the algorithm, we define some notations.
Let $w_1, w_2,\ldots, w_k$ be an ordering of vertices in $W$ such
that $h(w_1)$ is the last head encountered in a counterclockwise
direction from $h(v_1)$ in $F$ and $h(w_{i+1})$
immediately succeeds $h(w_{i})$ in a clockwise direction for $1\le
i \le k-1$. For a subset $S \subseteq V(G)$, we define $N[S] =
\setof{v \mid v \in N_G[u] \text{ and } u\in S}$. In addition, a
vertex $v$ in $V(G)-N[S]$ is said to be the {\em next undominated}
vertex with respect to $S$ if $t(v)$ is the first tail encountered
in a clockwise direction from $t(v_1)$.

Initially, the algorithm sets $S_i^t = \setof{w_i, p_t(w_i)}$ and
$S_i^h = \setof{w_i, p_h(w_i)}$. Then, it iteratively finds
the next undominated vertex $v$ with respect to
$S_i^t$ ($S_i^h$) and adds two vertices to $S_i^t$ ($S_i^h$) until
every vertex not in $S_i^t$ ($S_i^h$) is adjacent to a vertex in
$S_i^t$ ($S_i^h$). If $|S_i^t| \le |S_i^h|$, we have
$S_i = S_i^t$; otherwise, we have $S_i = S_i^h$. Finally, the algorithm selects a minimum paired-dominating set $S$ of $G$ from
$S_1,\ldots,S_k$ such that the cardinality of $S$ is the minimum. The steps of the algorithm are detailed below.

\medskip

\begin{algorithm}
\caption{Finding a minimum paired-dominating set in a circular-arc
graph}
\begin{algorithmic} [1]
\REQUIRE An intersection model $F$ of a circular-arc graph
$G$ with sorted endpoints.
\ENSURE A minimum paired-dominating set $S$ of
$G$.

\STATE let $W \leftarrow \{w \mid w \in N[v_1]$ and $a(w)$ is maximal\};
\STATE let $k\leftarrow |W|$;
\FOR  {each vertex $w_i \in W$}
\STATE let $S_i^t \leftarrow \setof{w_i, p_t(w_i)}$ and $S_i^h \leftarrow
\setof{w_i, p_h(w_i)}$;
\REPEAT
\STATE find the next undominated vertex $v$ with respect to $S_i^t$;
\STATE if $p_t(p_t(v)) \in S_i^t$, then let $S_i^t \leftarrow S_i^t \cup \setof{v,p_t(v)}$;
\STATE otherwise, let $S_i^t \leftarrow S_i^t \cup \setof{p_t(v),p_t(p_t(v))}$;
\UNTIL every vertex not in $S_i^t$ is adjacent to a vertex in $S_i^t$
\STATE repeat steps $5$ to $9$ to obtain $S_i^h$ by replacing $S_i^t$ with $S_i^h$;
\STATE if $|S_i^t| \le |S_i^h|$, then let $S_i \leftarrow S_i^t$; otherwise, let $S_i \leftarrow S_i^h$;
\ENDFOR \STATE choose $S$ from $S_1,\ldots,S_k$ such that the cardinality of $S$ is the minimum;
\RETURN $S$
\end{algorithmic}
\end{algorithm}


For the example in Figure~\ref{fig:2}, the
closed neighborhood of $v_1$ is $N_G[v_1] = \{v_1, v_2, v_{11}\}$
and we have $(w_1, w_2, w_3) = ( v_2, v_{11}, v_1)$. Then, by the
rules for finding $S^h_{i}$ and $S^t_{i}$ in Steps $3$ to $12$, we
have $S^h_{1} = \{v_2,v_{11},v_5,v_6,v_{9},v_{10}\}$, $S^t_{1} =
\{v_2,v_3,v_7,v_9\}$, $S^h_{2} =
\{v_{11},v_{10},v_5,v_6,v_{8},v_{9}\}$,  $S^t_{2} =
\{v_{11},v_2,v_5,v_6,v_{9},v_{10}\}$, $S^h_{3} =
\{v_{1},v_{11},v_5,v_6,v_{9},v_{10}\}$, and $S^t_{3} =
\{v_1,v_2,v_5,v_6,v_9,v_{10}\}$. Consequently, Step $11$ determines the sets $S_{1} =
\{v_2,v_3,v_7,v_9\}$, $S_{2} =
\{v_{11},v_2,v_5,v_6,v_{9},v_{10}\}$, and $S_{3} =
\{v_1,v_2,v_5,v_6,v_9,v_{10}\}$. Finally, Step
$13$ generates $S = \{v_2,v_3,v_7,v_9\}$, which is a
minimum paired-dominating set of $G$.

The properties of the following lemma are useful for finding a
minimum paired-dominating set and help us prove the correctness
of the algorithm.

\begin{lemma}\label{lemma:initial-pair}
Suppose $a(v)$ is a maximal arc in $F$ and $S_v$ is a minimum
paired-dominating set of $G$ among all the paired-dominating sets
that contains $v$. Then, there exists a minimum paired-dominating set $S_v$ such that we have $p_t(v) \in S_v$ or $p_h(v) \in
S_v$.
\end{lemma}
\begin{proof}
If $p_t(v) \in S_v$ or $p_h(v) \in S_v$, we are done;
otherwise, we assume that neither $p_t(v) \in S_v$ nor $p_h(v)
\in S_v$. Furthermore, let $v'$ be a vertex in $S_v$ such that a
perfect matching in $G[S_v]$ contains the edge $(v, v')$. Note
that $a(v)$ is a maximal arc in $F$. Hence, for each vertex $u \in
N_G(v)$, we have $N_G[v] \cup N_G[u] \subseteq N_G[v] \cup
N_G[p_t(v)]$ or $N_G[v] \cup N_G[u] \subseteq N_G[v] \cup
N_G[p_h(v)]$. For the case where $N_G[v] \cup N_G[v'] \subseteq
N_G[v] \cup N_G[p_t(v)]$, it is clear that $(S_v-\{v'\}) \cup
\{p_t(v)\}$ is a minimum paired-dominating set of $G$. Similarly,
for the case where $N_G[v] \cup N_G[v'] \subseteq N_G[v] \cup
N_G[p_h(v)]$, it isclear that $(S_v-\{v'\}) \cup \{p_h(v)\}$ is a
minimum paired-dominating set of $G$. The lemma then follows.
\end{proof}

\bigskip

Based on Lemma~\ref{lemma:initial-pair}, we are ready to prove
the following lemma, which provides the correctness of the
algorithm.

\begin{lemma}\label{lemma:Circular-arc-2}
Given an intersection model $F$ of a circular-arc graph $G$ with
sorted endpoints, Algorithm 2 outputs a minimum paired-dominating
set $S$ of $G$.
\end{lemma}
\begin{proof}
Clearly, $S$ is a paired-dominating set of $G$. To prove that
$S$ is a minimum paired-dominating set of $G$, it suffices to show
that, for each vertex $w_i \in W$, $S_i$ is a minimum paired-dominating set of $G$ among all the paired-dominating sets that contain $w_i$. According to Lemma~\ref{lemma:initial-pair}, there exists a minimum paired-dominating set $Z_i$ of $G$ among all the
paired-dominating sets containing $w_i$ such that we have
$p_t(w_i) \in Z_i$ or $p_h(w_i) \in Z_i$. Below, we only show that
the cardinality of $S_i$ is the minimum when $p_t(w_i) \in Z_i$. The
proof for the case where $p_h(w_i) \in Z_i$ is similar.

Let $S_i^t=\{v_{s_0}, v_{s_{0'}},\ldots, v_{s_x},v_{s_{x'}}\}$ be
a paired-dominating set of $G$ such that $(v_{s_0},
v_{s_{0'}})=(w_i,p_t(w_i))$, and let the vertices $v_{s_j}$ and
$v_{s_{j'}}$ be added to $S_i^t$ in the $j$th iteration of the
repeat-loop for $1 \le j \le x$. In addition, let $Z_i =\{v_{z_0}, v_{z_{0'}}, \ldots, v_{z_y},v_{z_{y'}}\}$ be a minimum paired-dominating set of $G$ such that $(v_{z_0}, v_{z_{0'}}), \ldots,
(v_{z_y},v_{z_{y'}})$ is a perfect matching of $G[Z_i]$ and
$z_{j-1} < z_{j}$ for $1 \le i \le y$; let $S_{ij}^t =
\{v_{s_0}, v_{s_{0'}},\ldots, v_{s_j},v_{s_{j'}}\}$; and $Z_{ij} =
\{v_{z_0}, v_{z_{0'}},\ldots, v_{z_j},v_{z_{j'}}\}$. For a subset
$R \subseteq V(G)$ and a vertex $w_i\in W$, we define that $A(R) = \max\{g \mid (1)~v_g \in R; \text{~or~} (2)~v_g \text{~is adjacent to a~}$ $\text{vertex~} v_j\in R \text{~and~} a(v_g)$ $\text{does not contain~} h(w_i)\}$. Hence, to prove that $S_i$ is a minimum
paired-dominating set of $G$ among all the paired-dominating sets
contain $w_i$, it suffices to show that $A(S_{ij}^t) \ge
A(Z_{ij})$ for $0 \le j \le y$. Clearly, the above statement can
be proved by induction on $j$. We omit the details of the proof
because they are similar to the arguments used to derive
Lemma~\ref{lemma:interval-correctness}.
\end{proof}

\medskip
\subsection{An $O(n)$-time implementation of Algorithm $2$
                              \label{section:implement-arc}}
The procedure for finding $S_i$ can be implemented in $O(n)$ time
by modifying Algorithm $1$ for interval graphs. Therefore, a naive
implementation of algorithm $2$ has a time complexity of $O(kn)$
with $k = |W|$. However, by exploiting the elegant properties of
circular-arc graphs, we can design a useful data structure that
helps us find all paired-dominating sets $S_1,S_2,\ldots,S_k$ in
$O(n)$ time. Hence, the time complexity of the algorithm for
finding a minimum paired-dominating set of a circular-arc graph
can be improved to $O(n)$.

To find a minimum paired-dominating set $S$ in $G$, algorithm $2$
first obtains the paired-dominating sets
$S^t_1,S^t_2,\ldots,S^t_k$ and $S^h_1,S^h_2,\ldots,S^h_k$. Then,
the set $S$ with the minimum cardinality is chosen from the sets.
In the following, we only consider an $O(n)$-time implementation
for finding all the paired-dominating sets
$S^t_1,S^t_2,\ldots,S^t_k$. Using a similar method, it can can be
shown that $S^h_1,S^h_2,\ldots,S^h_k$ can also be obtained in
$O(n)$ time. For simplicity, we denote $S^t_i$ and $p_t(v)$ by
$S_i$ and $p(v)$ respectively, in the remainder of this section.

For each vertex $w_i \in W$, the algorithm constructs a
paired-dominating set $S_i$ containing $w_i$ and $p(w_i)$. Suppose
that $S_i = \{w_{i_{0}}, w_{i_{0'}},\ldots,
w_{i_{\ell}},w_{i_{\ell'}}\}$, where $(w_{i_{0}},
w_{i_{0'}})=(w_i,p(w_i))$, and that the vertices $w_{i_{j}}$ and
$w_{i_{j'}}$ are added to $S_i$ in the $j$th iteration of the
repeat-loop for $1 \le j \le \ell$. In addition, let
$S_x=\{w_{x_{0}}, w_{x_{0'}},\ldots, w_{x_{a}},w_{x_{a'}}\}$ and
$S_y=\{w_{y_{0}}, w_{y_{0'}},\ldots, w_{y_{b}},w_{y_{b'}}\}$ be
two paired-dominating sets of $G$ such that $S_x, S_y \in
\setof{S_1,S_2,\ldots,S_k}$. It is clear that if $(w_{x_{e}},
w_{x_{e'}})= (w_{y_{f}}, w_{y_{f'}})$, then $(w_{x_{(e+1)}},
w_{x_{(e+1)'}}) = (w_{y_{(f+1)}}, w_{y_{(f+1)'}}),(w_{x_{(e+2)}},
w_{x_{(e+2)'}})= (w_{y_{(f+2)}}, w_{y_{(f+2)'}}),\ldots,$
$(w_{x_{(e+g)}}, w_{x_{(e+g)'}})= (w_{y_{(f+g)}},
w_{y_{(f+g)')}}$, where $g = \min \setof{a-e,b-f}-1$. Based on the
above observation, we define a digraph $D$ to improve the
complexity of the algorithm from $O(kn)$ to $O(n)$.

\vspace{14pt}

\begin{figure}[htb]
\centerline{\input{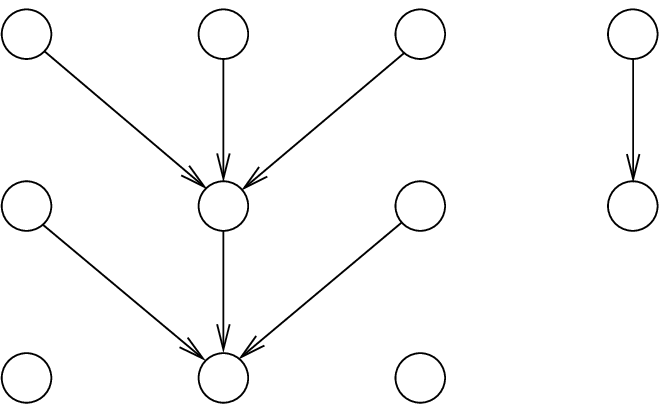}} \caption{The corresponding digraph $D$
for the circular-arc graph $G$ in Figure~\ref{fig:2}$(b)$.}
\label{fig:3}
\end{figure}

Recall that a vertex $v$ in $V(G)-N[S]$ is deemed the next
undominated vertex with respect to a subset $S \subseteq V(G)$ if
$t(v)$ is the first tail encountered in a clockwise direction from
$t(v_1)$. We define $succ((v,p(v))) = (p(u),p(p(u)))$ for
$v\in V(G)$, where $u$ is the next undominated vertex with respect
to $\setof{v,p(v)}$. Let $D = (V(D), E(D))$ be a digraph such that

\vspace{-30pt}
\begin{align*}
V(D) &= \setof{(v,p(v))\mid v\in V(G)}; \text{~and}\\
E(D) &= \{((v,p(v))\rightarrow (u,p(u)) \mid (u,p(u)) = succ(v,p(v)) \text{ and } u, p(u)\not \in N[v_1]\}.
\end{align*}

\noindent Figure~\ref{fig:3} shows an example of the corresponding
digraph $D$ for the circular-arc graph $G$ in
Figure~\ref{fig:2}$(b)$.

Next, we show that the graph $D$ is a directed forest graph that
can be constructed in $O(n)$ time. By the definition of $E(D)$,
there exists no vertex $(v, p(v))$ in $D$ such that the in-degree
of $(v, p(v))$ is greater than or equal to $1$; and
$\setof{v,p(v)} \cap N[v_1] \not = \emptyset$. Meanwhile, because
every cycle $C = (c_1, c_2, \ldots, c_q) $ in $D$ must contain
such a vertex $c_i = (v, p(v))$, $D$ is a directed forest graph.
The $O(n)$-time procedure used to construct $D$ is as follows.
Using similar arguments to those presented in
Section~\ref{section:implement-int}, it can be shown that, with
$O(n)$-time preprocessing, $p(v)$ and $succ(v,p(v))$ can be
determined in $O(1)$ time for each vertex $v \in V(G)$. Moreover,
because the out-degree of each vertex in $V(D)$ is at most one and
$|V(D)| = n$, the digraph $D$ can be constructed in $O(n)$ time.

Let $P_i$ denote the maximal directed path in $D$ starting from
$(w_i, p(w_i))$ for each vertex $w_i \in W$, and let $\bar{P_i} =
\setof{v,u \mid (v,u)\in V(P_i)}$. The next lemma provides an
important property that can be used to derive the
paired-domination sets $S_i$ from maximal directed path $P_i$ for $1 \le i \le k$.

\begin{lemma}\label{lemma:implementation-2}
Suppose that $P_i$ is the maximal directed path in $D$ starting
from $(w_i, p(w_i))$ and $\bar{P_i} = \setof{v,u \mid (v,u)\in
V(P_i)}$. Then, we have $\bar{P_i} \subseteq S_i$ and
$|S_i|-|\bar{P_i}| \le 4$ for $1 \le i \le k$.
\end{lemma}
\begin{proof}
Because $D$ is a directed forest graph, it is clear from the definitions of $S_i$ and $P_i$ that  $\bar{P_i} \subseteq S_i$. Suppose $\bar{P_i}  = \{w_{i_{0}}, w_{i_{0'}},\ldots,
w_{i_{\ell}},w_{i_{\ell'}}\}$ such that $(w_{i_{0}},
w_{i_{0'}})=(w_i,p(w_i))$ and $(w_{i_{j}},w_{i_{j'}}) =
succ((w_{i_{(j-1)}},w_{i_{(j-1)'}}))$ for $1 \le j \le \ell$.
Then, by the definition of $P_i$, we can verify that the vertices
in $V(G) - N[v_1]$ are dominated by $\bar{P_i} \cup
\setof{w_{i_{(\ell+1)}},w_{i_{(\ell+1)'}}}$, where
$(w_{i_{(\ell+1)}},w_{i_{(\ell+1)'}}) =
succ((w_{i_{\ell}},w_{i_{\ell'}}))$. This implies that we have
$|S_i|-|\bar{P_i}| \le 4$ as desired.
\end{proof}

\medskip

In the following, we show the paired-domination sets $S_1, S_2,\ldots, S_k$ can be obtained in $O(n)$ time by exploiting Lemma~\ref{lemma:implementation-2}. First, we determine whether or not the set $\bar{P_i}$ is a paired-domination set of $G$ for $1 \le i \le k$. If the answer is positive, we set $S_i = \bar{P_i}$. Otherwise, we set $S_i = \bar{P_i}$ and augment $S_i$ with the vertices $p(u)$ and $p(p((u))$ until $S_i$ becomes a
paired-domination set of $G$, where $u$ is the next undominated
vertex with respect to $S_i$. According to
Lemma~\ref{lemma:implementation-2}, the augmentation will occur twice at most. This implies that the augmentation can be
completed in $O(1)$ time for each vertex $w_i \in W$. Furthermore,
because $\bar{P_i} \subseteq S_i$ and the digraph $D$ can be
constructed in $O(n)$ time, the paired-dominating sets
$S_1,S_2\ldots,S_k$ can be obtained in $O(n)$ time. Now that $D$
is a directed forest, the length of $P_i$ in $D$ can be determined
in $O(n)$ time by running a depth first search algorithm on all
vertices $(v, p(v))$ in $\bar{D} = (V(\bar{D}), E(\bar{D}))$ such
that the in-degree of $(v, p(v))$ is equal to $0$, where
$V(\bar{D}) = V(D)$ and $E(\bar{D}) = \setof{ ((v,p(v))\rightarrow
(u,p(u)) \mid ((u,p(u))\rightarrow (v,p(v)) \in E(D)} $. It
follows that a minimum paired-dominating set $S$ of $G$ can be
chosen from $S_1,\ldots,S_k$ in $O(n)$ time.

Combining Lemma~\ref{lemma:Circular-arc-2} and above discussion,
we have the following theorem.

\begin{theorem} \label{theorem:interval} Given an intersection
model $F$ of a circular-arc graph $G$ with sorted endpoints,
algorithm $2$ outputs a minimum paired-dominating set $S$
of $G$ in $O(n)$ time.
\end{theorem}

\section{Concluding Remarks  \label{section:conclusion}}
We have proposed two algorithms for the paired-domination problem in interval graphs and circular-arc graphs respectively. The algorithm for interval graphs produces a minimum paired-dominating set incrementally in a greedy manner. We extended the results to design the algorithm for circular-arc graphs. If the input graph is comprised of a family of $n$ arcs, both algorithms can be implemented in $O(n \log n)$ time. However, if the endpoints of the arcs are sorted, both algorithms only require $O(n)$ time. These results are optimal within a constant factor.

Finally, we consider some open questions related to the paired-domination problem. It would be interesting to investigate the weighted analogue of this problem, i.e., to compute a minimum weight paired-dominating set in which each vertex is associated with a weight. Furthermore, many optimization problems are NP-complete if they are defined on general graphs; and they are solvable in polynomial time if they are defined on some special classes of graphs, such as bounded treewidth graphs, co-comparability graphs, and distance-hereditary graphs. Therefore, it would also be interesting to design polynomial-time algorithms for these graph classes. In addition, it would be useful if we could develop a polynomial-time approximation algorithm for general graphs; or prove that the problem remains NP-complete in planar graphs and devise a polynomial-time approximation scheme for it.

\small
\bibliographystyle{abbrv}
\bibliography{PDom}
\end{document}